\documentclass[runningheads]{llncs}
\usepackage[T1]{fontenc}
\usepackage{graphicx}
\usepackage{tikz}
\usepackage{mathtools}
\usepackage{xcolor}

\definecolor{Yellow}{RGB}{255, 241, 0}

\definecolor{colorCrossWE}{RGB}{21, 142, 179}
\definecolor{colorCrossNS}{RGB}{230, 115, 0}

\definecolor{colorBgCrossWE}{RGB}{27, 182, 230}
\definecolor{colorBgCrossNS}{RGB}{255, 128, 0}

\usetikzlibrary{matrix,positioning,arrows,arrows.meta,fit,bending,shapes.misc}

\tikzset{
  ain/.style={<-,thick},
  aout/.style={->,thick},
  celltypedA/.style={
    fill=colorBgCrossWE
  },
  celltypedB/.style={
    fill=colorBgCrossNS,
    font=\bfseries
  },
  neigArr/.style={
    -{Stealth}_,
    thick
  },
  firingLabel/.style={
    thick,
    red,
    font=\bfseries
  },
  farr/.style={
    -{Stealth}_,
    thick
  },
  farrA/.style={
    farr,colorCrossWE
  },
  farrB/.style={
    farr,colorCrossNS,densely dashed
  },
  crossArr/.style={
    farr,colorCrossNS
  },
  cellHighlight/.style={
    fill=lime,
    font=\bfseries
  },
  table/.style={
    matrix of nodes,
    row sep=-\pgflinewidth,
    column sep=-\pgflinewidth,
    nodes={rectangle,draw=black,align=center,minimum size=7mm},
    text depth=0.25ex,
    text height=2.5mm,
    nodes in empty cells
  },
}
\usepackage{amsmath}
\usepackage{amssymb}
\usepackage{stmaryrd}
\usepackage{xspace}

\graphicspath{{figures/}}

\def\nei{\mathcal{N}}
\def\prim{\mathcal{P}}

\def\pred{\textbf{PRED}\xspace}

\def\MCV{\textbf{MCVP}\xspace}

\def\P{{\fontfamily{cmss}\selectfont P}\xspace}

\def\NC{{\fontfamily{cmss}\selectfont NC}\xspace}

\usepackage[most]{tcolorbox}
\newtcolorbox{problemBox}[1][]{
  enhanced,
  breakable,
  colback=white,
  colbacktitle=white,
  coltitle=black,
  fonttitle=\bfseries,
  boxrule=.6pt,
  titlerule=.2pt,
  toptitle=3pt,
  bottomtitle=3pt,
  left=6pt,
  right=6pt,
#1}

\usepackage{cite}
\usepackage[hidelinks]{hyperref}

\newcommand{\neidir}[1]{\stackrel{#1}{\to}}

\begin{document}
\title{Non-Uniform and Weighted Crossing Gates in Two-Dimensional Sandpiles}
%
%
\author{Pablo Concha-Vega\inst{1}\orcidID{0009-0001-2419-1687} \and
  Antonin Loubière\inst{1, 2}\orcidID{0009-0006-5320-6776}   \and
Kévin Perrot\inst{1}}
\authorrunning{P. Concha-Vega et al.}
%
\institute{Aix Marseille Université, CNRS, LIS, Marseille, France \and
  École normale supérieure Paris-Saclay, Université Paris-Saclay, France}
\maketitle              
\begin{abstract}
  Determining whether predicting two-dimensional sandpiles lies in $\NC$
  or is \P-complete has been open for decades.
  Moore and Nilsson proved \P-completeness
  for the three dimensional case by encoding Boolean circuits
  into sandpiles, but this method fails
  in two dimension due to the impossibility of crossing gates.

  In this work, we study the existence of crossing gates on
  non-uniform and weighted grids.
  We establish an equivalence between uniform weighted crossing gates
  and a class of simple non-uniform crossing gates, which we call primal.
  We also exhibit a crossing gate that inherently requires more than one crossing,
  rather than a single crossing as in standard constructions.
  Finally, we show that the equivalence between uniform weighted and
  primal crossings breaks down in more general settings.

\end{abstract}

\section{Introduction}

The concept of self-organized criticality was introduced
in~\cite{bak1987self} with the sandpile model as its paradigmatic example.
Beyond its physical relevance, sandpiles have been extensively studied
under the lens of complexity theory as a type of cellular automaton
\cite{moore_computational_1999,MEJIA20113964,miltersen2007computational}.
A central question in this context is how hard it is to compute
the future global states of the model.
This is captured by the \emph{prediction problem}: given a distribution
of sand grains and a target cell, does the target cell eventually topple?
Early results by Moore and Nilsson~\cite{moore_computational_1999} show
a dichotomy: in one-dimensional sandpiles, the problem can be decided
efficiently in parallel (it belongs to $\text{\NC}^3$); while for three or more
dimensions, it becomes a \P-complete problem, leaving the two-dimensional
case open and remains open to this day.
The main issue is that \P-completeness proofs
(in this context) rely on embedding non-planar monotone Boolean circuits
on the sandpile grid, therefore two dimensions is intuitively not enough.
This was formalized by Gajardo and Goles~\cite{gajardo_crossing_2006} by
proving that no \emph{crossing gate} is possible on two-dimensional
sandpile (with the von Neumann or Moore neighborhoods). Crossing gates
simulate two signals crossing, enabling connections between other
circuit components. Crossing gates are therefore
crucial when proving \P-completeness in the context of two-dimensional
cellular automata \cite{banks1971information,goles2018complexity,concha2025sandpiles,concha2022complexity} and are typically
the hardest component to construct.
The latter motivated a line of research focusing on embedding
arbitrary Boolean circuits into two-dimensional sandpiles under
different settings
\cite{formenti2012computational,nguyen_any_2018,modanese2024embedding,concha2025sandpiles,concha2025timed}.
This work is intended to contribute to this line of research
with a particular focus on crossing gates.
\paragraph{Contribution.}
Our contribution consists in three main results.
First, we define the notions of primal and weighted grids.
In a $\prim$-primal grid, each cell has at most one out-neighbor from $\prim$.
In a weighted $\prim$-grid, the neighborhood is defined as a multiset
over $\prim$, therefore the distribution of grains to the neighbors
can differ from the uniform distribution, but is identical across all cells.
We prove that a $\prim$-primal grid admits a crossing gate
if and only if there exists a weighted $\prim$-grid that allows a crossing gate.
This result tackles one of the research directions proposed by Concha-Vega
et al.~\cite{concha2025sandpiles} related to \emph{non-uniform sandpiles}.

Secondly, we ask whether all crossing gates have exactly two edges
(of their toppling graphs) that cross, motivated by observation of
the known crossing gates in the literature. We give a negative
answer by showing a neighborhood that always needs more crossings edges.

Finally, we show that the equivalence fails for general decompositions (in 
non-primal and non-uniform grids), by giving counter-examples for both implications.

\paragraph{Outline}

This article is organized as follows.
In Section~\ref{sec:pre}, we give an overview of the necessary concepts and
previous results relevant to our work.
Our results are then presented in Sections~\ref{sec:primal_iff_weighted},~\ref{sec:imbricated},
and,~\ref{sec:general-decomposition} respectively.
Finally, we conclude in Section~\ref{sec:conclusions}.

\section{Preliminaries}\label{sec:pre}

\subsection{Sandpile automaton}

A two dimensional sandpile is a cellular automaton defined on a $\ell \times \ell'$
grid, where each cell contains a nonnegative number of sand grains.
We denote the set of cells by $V := \llbracket 1, \ell \rrbracket \times
\llbracket 1, \ell' \rrbracket$. A configuration is a function $c: V \rightarrow
\mathbb{N}$ that assign to each cell the number of grains it contains. When
a cell has too many grains, it \textit{topples} and sends some of its
grains to neighboring cells; this defines the dynamics of the system.
Each cell $u \in V$ is associated with a neighborhood
$\nei_u \subseteq \mathbb{Z}^2 \backslash \{(0, 0)\}$.
A cell is \textit{stable} if it contains less than $|\nei_u|$ grains;
otherwise, it is \textit{unstable} and topples, sending one grain to each
of its neighbors $v = u + \delta$ where $\delta \in \nei_u$.
We extend the notion of stability to configurations by saying that a
configuration is \textit{stable} if all its cells are stable.
Formally, applying one parallel step of this rule to a configuration
defines the following update function $T$:

\[ T(c)(u) = c(u) - |\nei_u| H(c(u) - |\nei_u|) + \sum_{v \in V}
\nei_v (u - v) H(c(v) - |\nei_{v}|) \]

\begin{figure}[t]
  \centering
  \begin{tikzpicture}
    \input{tikzifieur/neig-ex.tex}
  \end{tikzpicture}
  \caption{The von Neumann neighborhood (left) and the Moore
  neighborhood (right).}
  \label{fig:von-neumann-moore-nei}
\end{figure}

\begin{figure}[h!]
  \centering
  \scalebox{0.65}{
    \begin{tikzpicture}
      \input{tikzifieur/example.tex}
    \end{tikzpicture}
  }
  \caption{An example of stabilization in 3 steps with the von
  Neumann neighborhood.}
  \label{fig:stabilization-example}
\end{figure}

Where $H(x) :=
1$ if $x \geq 0$ and $0$ if $x<0$,
is the \emph{Heaviside step} function.
Most works consider that all cells have the same neighborhood;
we call such grids \textit{uniform}.
This neighborhood is typically assumed to span $\mathbb{Z}^2$. 
We allow $\nei_u$ to be a multiset i.e., a toppling cell may send multiple grains to the same neighbor.
Grids with multiset neighborhoods are called
\textit{weighted} (see Figure~\ref{fig:crossing-moore} left).

To simplify the study of uniform sandpiles, we introduce a sink
vertex denoted $s$, which never topples.
Whenever a cell has neighbors outside of the grid boundaries,
it sends the corresponding grains to $s$.
Likewise, any neighbors that would have sent grains from outside
the grid are redirected to $s$.
Under this convention, each cell has as many out-neighbors as
in-neighbors, so uniform grids can be assumed to be Eulerian.

Two uniform grids have been greatly studied in the literature
\cite{dhar_self-organized_1990,gajardo_crossing_2006,moore_computational_1999}: the
uniform grid with the von Neumann neighborhood, the four cardinal
adjacent cells, and with the Moore neighborhood, the eight adjacent cells
(Figure~\ref{fig:von-neumann-moore-nei}).

Dhar has shown that configurations on uniform grids with von Neumann
or Moore neighborhoods can always be stabilized by
applying $T$ a polynomial number of times,
and that a non-deterministic sequential application of topplings (one-by-one)
does not impact the resulting
stable configuration \cite{dhar_self-organized_1990}.
This result can be extended to any uniform grid: a configuration
on a uniform grid can always be stabilized in a polynomial number of topplings
($\#\text{grains} \cdot \ell^4$). Figure \ref{fig:stabilization-example}
presents an example of stabilization.

\subsection{Toppling prediction and crossings}

We focus on the decision problem of predicting
whether a given cell will topple during the stabilization of a
configuration as formalized below. As non-uniform grid are not always
stabilizable, and stabilization may not occur in polynomial time,
we restrict this problem to uniform grids.

\begin{problemBox}[title=$\nei$-Prediction Problem ($\nei$-\pred)]
  \paragraph{Input:} \parbox[t]{.85\textwidth}{
    A configuration $c$ on the uniform grid
    generated by $\nei$, and a target cell $y \in V$.
  }\\[.5em]
  \paragraph{Output:} Decide whether $y$ topples during the
  stabilization of $c$.
\end{problemBox}

Moore and Nilsson \cite{moore_computational_1999} proved that
\pred is \P-complete for
3-dimensional sandpiles with the von Neumann or Moore neighborhoods.
In contrast, for 1-dimensional sandpiles, the prediction problem
is in \NC (polylogarithmic time with a polynomial number of processors)
for the von Neumann neighborhood.

To prove the \P-completeness of 3-dimensional \pred, they use a
reduction from the \textsc{Monotone Circuit Value Problem}. The
circuit is embedded in a sandpile configuration using gadgets that implement,
the OR and AND operations. These gadgets are connected by wires,
which consist of sequences of stable cells holding a maximal number of grains.
A true signal is represented as the chain reaction of toppling along a wire:
when a cell topples, it destabilizes the next cell, thereby
propagating the signal.

\begin{problemBox}[title=Monotone Circuit Value Problem (\MCV)]
  \paragraph{Input:} \parbox[t]{.85\textwidth}{
    A Boolean circuit composed of input nodes, an output node,
    AND gates, and OR gates. A valuation $v$ of the input gates.
  }\\[.5em]
  \paragraph{Output:}
  Decide whether the circuit outputs true when evaluated under $v$.
\end{problemBox}

However, their construction does not apply for 2-dimensional sandpiles because the
non-planarity of the circuits is required for \P-completeness
\cite{ramachandran_efficient_1996}.
Moreover, Gajardo and Goles \cite{gajardo_crossing_2006} showed
that 2-dimensional uniform grids with von Neumann or Moore neighborhoods
cannot implement signal crossings using toppling chains
(and no other encoding of circuit is currently known).
Formally, they proved that no \emph{crossing gate} exists in
2-dimensional Moore or von Neumann grids.
Intuitively, a crossing gate is a gadget that allows toppling wires
to cross without interfering with each other.

\begin{definition}[Crossing gate \cite{gajardo_crossing_2006}]
  A \emph{crossing gate} can be from any direction, but, by symmetry and rotation,
  it suffices to consider the case where one wire runs from North to South
  and the other from West to East. For clarity, we will use theses directions for
  this definition.

  A \emph{crossing gate} is a stable configuration $c$ with the following property.
  If a grain is added to a designated cell on the North edge
  (excluding the corners), a designated cell on the South edge topples,
  but no cell on the East edge topples.
  In this case, we say that $c$ transports from North to South
  independently of the East.
  Similarly, $c$ must transport from West to East
  independently of the South.
\end{definition}

A more precise definition can be found in~\cite{concha2025sandpiles}.
When studying crossing gates, it is useful to consider the
\emph{toppling graphs} of the crossing gate, which represent the order
in which topplings occur justifying that this parallel definition will also capture sequential dynamics.
Figure~\ref{fig:toppling-graph-example}
presents an example of a simple crossing gate with its two toppling graphs.

\begin{definition}[Toppling graph \cite{gajardo_crossing_2006}]
  The \emph{North-toppling graph} (resp.~\emph{West-toppling graph}) of a
  crossing gate $c$ is the directed graph whose
  vertex set consists of all the cells that topple when a
  grain is added to the designated North (resp.~West) side cell.
  There is a directed edge from $u$ to $v$ if and only of
  $v$ is a neighbor of $u$ ($v - u \in \nei_u$)
  and $v$ topples after $u$.
\end{definition}

\begin{figure}[t]
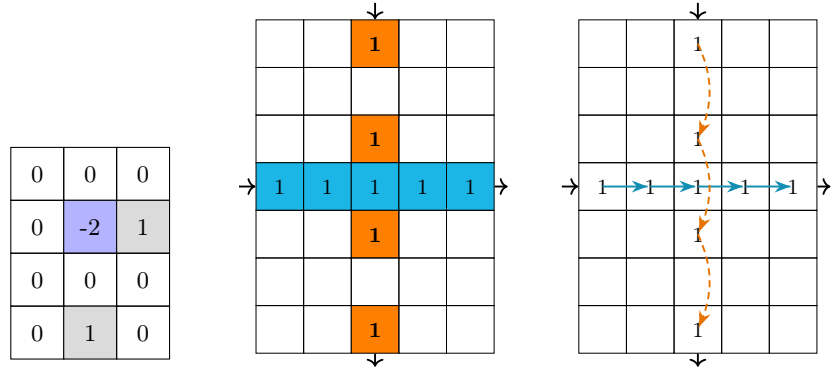

  \centering
  \begin{tikzpicture}
    \input{tikzifieur/inter-simple-neig.tex}
  \end{tikzpicture}
  \scalebox{0.9}{
    \qquad
    \begin{tikzpicture}
      \input{tikzifieur/inter-simple-crossing.tex}
    \end{tikzpicture}
  }
  \caption{A neighborhood (left), a crossing gate (center) and its toppling graphs (right).}
  \label{fig:toppling-graph-example}
\end{figure}

One important fact that justifies the previous definition is that,
when a single grain is added to each cell along an edge,
no cell topples more that once.
Gajardo and Goles proved this fact for von Neumann and Moore neighborhoods,
but the result holds for all Eulerian grid.
To see why, consider the first cell to topple twice.
It cannot be an interior cell, because initially it is stable and
can receive at most $|\nei_u|$ grains if all its incoming
neighbors have toppled. The designated cell initially receive one additional
grain, but since it has an incoming edge from the sink $s$,
which never topples, it cannot topple twice either.

Gajardo and Goles~\cite{gajardo_crossing_2006} proved that
intersections between the toppling graphs are useless and
that crossing gates can always be constructed such that
the North-toppling and West-toppling graphs have disjoint vertex sets.
This result was later extended to Eulerian grids by Nguyen
and Perrot \cite{nguyen_any_2018}.

In the next sections, we will present our contribution about the link between
non-uniform and weighted crossing gates.

\section{Trading primal non-uniformity for weighted uniformity}
\label{sec:primal_iff_weighted}

In this section, we present our main result:
a primal non-uniform grid admits a crossing gate if and only if
there exists a weighting of these neighborhoods that
admits a crossing gate on a uniform grid.

\begin{definition}[$\prim$-primal non-uniform grid]
  Let $\prim \subseteq \mathbb{Z}^2$ be a finite set.
  A \emph{$\prim$-primal non-uniform grid} is
  a grid in which, for every cell $v$, the neighborhood $\nei_v$
  satisfies $|\nei_v| = 1$ and $\nei_v \subseteq \prim$.
  In other words, each cell has exactly one out-neighbor
  from $\prim$.
  The uniform grid with neighborhood $\prim$ is called
  $\prim$-grid.
\end{definition}

\paragraph{Remark.}
If $|\prim| = 1$, then the grid is uniform and the model is trivial.
Hence, we are mainly interested in the case $|\prim| > 1$, which we
simply call $\prim$-primal grids.

The existence of a crossing gate on a $\prim$-primal grid is closely
related to the planarity of the $\prim$-grid. Indeed, if there
is a $\prim$-primal grid that admits a crossing gate, as it transports
grains from the North to the South, and from the West to the East distinctly,
then the $\prim$-grid is necessarily non-planar.
Conversely, if a $\prim$-grid is non-planar,
we can use the directions that cross to create a crossing gate
on a $\prim$-primal grid.
The fact that each cell has only one neighbor will ensure that the configuration
transports grains independently.
The last notion we need in order to state our theorem
is that of a weighted neighborhood.

\begin{definition}[Weighted $\prim$-grid]
We define a \emph{weighted $\prim$-grid}
as a uniform grid with neighborhood $\nei$,
where $\nei$ is a multiset over $\prim$
(that is, elements of $\prim$ may appear multiple times),
and every element of $\prim$ appears at least once in
$\nei$, i.e., $\prim \subseteq \nei$.
\end{definition}

\begin{figure}[t]
  \centering
  \scalebox{0.7}{
    \tikzset{every picture/.style={line width=0.75pt}} 

\begin{tikzpicture}[x=0.75pt,y=0.75pt,yscale=-1,xscale=1]

  \draw [color=colorCrossWE  ,draw opacity=1 ]   (261.4,185.33) -- (302.35,213.45) ;
  \draw [shift={(304,214.58)}, rotate = 214.47] [color=colorCrossWE  ,draw opacity=1 ][line width=0.75]    (10.93,-3.29) .. controls (6.95,-1.4) and (3.31,-0.3) .. (0,0) .. controls (3.31,0.3) and (6.95,1.4) .. (10.93,3.29)   ;
  \draw [color=colorCrossNS  ,draw opacity=1 ] [dash pattern={on 4.5pt off 4.5pt}]  (297.5,184.58) -- (252.61,217.89) ;
  \draw [shift={(251,219.08)}, rotate = 323.43] [color=colorCrossNS  ,draw opacity=1 ][line width=0.75]    (10.93,-3.29) .. controls (6.95,-1.4) and (3.31,-0.3) .. (0,0) .. controls (3.31,0.3) and (6.95,1.4) .. (10.93,3.29)   ;
  \draw [color=colorCrossWE  ,draw opacity=1 ]   (167.8,247.08) .. controls (236.11,287.18) and (182.4,145.95) .. (260.4,184.85) ;
  \draw [shift={(262.8,186.08)}, rotate = 207.82] [fill=colorCrossWE  ,fill opacity=1 ][line width=0.08]  [draw opacity=0] (10.72,-5.15) -- (0,0) -- (10.72,5.15) -- (7.12,0) -- cycle    ;
  \draw [color=colorCrossWE  ,draw opacity=1 ]   (304.2,213.08) .. controls (372.51,272.98) and (379.07,144.2) .. (419.95,209.54) ;
  \draw [shift={(421.2,211.58)}, rotate = 239.04] [fill=colorCrossWE  ,fill opacity=1 ][line width=0.08]  [draw opacity=0] (10.72,-5.15) -- (0,0) -- (10.72,5.15) -- (7.12,0) -- cycle    ;
  \draw [color=colorCrossNS  ,draw opacity=1 ] [dash pattern={on 4.5pt off 4.5pt}]  (251,219.08) .. controls (277.1,268.83) and (330.37,220.09) .. (312.37,265.92) ;
  \draw [shift={(311.5,268.08)}, rotate = 292.5] [fill=colorCrossNS  ,fill opacity=1 ][line width=0.08]  [draw opacity=0] (10.72,-5.15) -- (0,0) -- (10.72,5.15) -- (7.12,0) -- cycle    ;
  \draw [color=colorCrossNS  ,draw opacity=1 ] [dash pattern={on 4.5pt off 4.5pt}]  (298.67,138.92) .. controls (335.17,125.92) and (418.17,244.42) .. (384.67,265.92) .. controls (351.5,287.2) and (345.29,120.31) .. (301.99,179.55) ;
  \draw [shift={(300.67,181.42)}, rotate = 304.6] [fill=colorCrossNS  ,fill opacity=1 ][line width=0.08]  [draw opacity=0] (10.72,-5.15) -- (0,0) -- (10.72,5.15) -- (7.12,0) -- cycle    ;

  \draw (139.17,229.17) node [anchor=north west][inner sep=0.75pt]   [align=left] {\textbf{{\LARGE W}}};
  \draw (418.5,215.83) node [anchor=north west][inner sep=0.75pt]   [align=left] {\textbf{{\LARGE E}}};
  \draw (300.67,269.17) node [anchor=north west][inner sep=0.75pt]   [align=left] {\textbf{{\LARGE S}}};
  \draw (291.5,101) node [anchor=north west][inner sep=0.75pt]   [align=left] {\textbf{{\LARGE N}}};
  \draw (191.5,185.83) node [anchor=north west][inner sep=0.75pt]   [align=left] {$\displaystyle a$};
  \draw (427.17,191.17) node [anchor=north west][inner sep=0.75pt]   [align=left] {$\displaystyle b$};
  \draw (350,137.67) node [anchor=north west][inner sep=0.75pt]   [align=left] {$\displaystyle c$};
  \draw (259.67,249.33) node [anchor=north west][inner sep=0.75pt]   [align=left] {$\displaystyle d$};
  \draw (380,273) node [anchor=north west][inner sep=0.75pt]  [color=colorCrossNS  ,opacity=1 ] [align=left] {$\displaystyle C_{1}$};
  \draw (184,265) node [anchor=north west][inner sep=0.75pt]  [color=colorCrossWE  ,opacity=1 ] [align=left] {$\displaystyle C_{2}$};
  \draw (276,211) node [anchor=north west][inner sep=0.75pt]  [color=colorCrossWE  ,opacity=1 ] [align=left] {$\displaystyle \vec{U}$};
  \draw (283,150) node [anchor=north west][inner sep=0.75pt]  [color=colorCrossNS  ,opacity=1 ] [align=left] {$\displaystyle \vec{V}$};

\end{tikzpicture}
    \tikzset{every picture/.style={line width=0.75pt}} 

\begin{tikzpicture}[x=0.75pt,y=0.75pt,yscale=-0.8,xscale=0.8]

  \draw [color=colorCrossWE  ,draw opacity=1 ]   (261.4,185.33) -- (302.35,213.45) ;
  \draw [shift={(304,214.58)}, rotate = 214.47] [color=colorCrossWE  ,draw opacity=1 ][line width=0.75]    (10.93,-3.29) .. controls (6.95,-1.4) and (3.31,-0.3) .. (0,0) .. controls (3.31,0.3) and (6.95,1.4) .. (10.93,3.29)   ;
  \draw [color=colorCrossNS  ,draw opacity=1 ] [dash pattern={on 4.5pt off 4.5pt}]  (297.5,184.58) -- (252.61,217.89) ;
  \draw [shift={(251,219.08)}, rotate = 323.43] [color=colorCrossNS  ,draw opacity=1 ][line width=0.75]    (10.93,-3.29) .. controls (6.95,-1.4) and (3.31,-0.3) .. (0,0) .. controls (3.31,0.3) and (6.95,1.4) .. (10.93,3.29)   ;
  \draw [color=colorCrossWE  ,draw opacity=1 ]   (123.8,217.08) .. controls (192.11,257.18) and (138.4,115.95) .. (216.4,154.85) ;
  \draw [shift={(218.8,156.08)}, rotate = 207.82] [fill=colorCrossWE  ,fill opacity=1 ][line width=0.08]  [draw opacity=0] (10.72,-5.15) -- (0,0) -- (10.72,5.15) -- (7.12,0) -- cycle    ;
  \draw [color=colorCrossWE  ,draw opacity=1 ]   (389.2,273.08) .. controls (457.51,332.98) and (464.07,204.2) .. (504.95,269.54) ;
  \draw [shift={(506.2,271.58)}, rotate = 239.04] [fill=colorCrossWE  ,fill opacity=1 ][line width=0.08]  [draw opacity=0] (10.72,-5.15) -- (0,0) -- (10.72,5.15) -- (7.12,0) -- cycle    ;
  \draw [color=colorCrossNS  ,draw opacity=1 ] [dash pattern={on 4.5pt off 4.5pt}]  (205,258.08) .. controls (231.1,307.83) and (284.37,259.09) .. (266.37,304.92) ;
  \draw [shift={(265.5,307.08)}, rotate = 292.5] [fill=colorCrossNS  ,fill opacity=1 ][line width=0.08]  [draw opacity=0] (10.72,-5.15) -- (0,0) -- (10.72,5.15) -- (7.12,0) -- cycle    ;
  \draw [color=colorCrossNS  ,draw opacity=1 ] [dash pattern={on 4.5pt off 4.5pt}]  (345.67,103.92) .. controls (382.17,90.92) and (465.17,209.42) .. (431.67,230.92) .. controls (398.5,252.2) and (392.29,85.31) .. (348.99,144.55) ;
  \draw [shift={(347.67,146.42)}, rotate = 304.6] [fill=colorCrossNS  ,fill opacity=1 ][line width=0.08]  [draw opacity=0] (10.72,-5.15) -- (0,0) -- (10.72,5.15) -- (7.12,0) -- cycle    ;
  \draw [color=colorCrossNS  ,draw opacity=1 ] [dash pattern={on 4.5pt off 4.5pt}]  (251,219.08) -- (206.11,252.39) ;
  \draw [shift={(204.5,253.58)}, rotate = 323.43] [color=colorCrossNS  ,draw opacity=1 ][line width=0.75]    (10.93,-3.29) .. controls (6.95,-1.4) and (3.31,-0.3) .. (0,0) .. controls (3.31,0.3) and (6.95,1.4) .. (10.93,3.29)   ;
  \draw [color=colorCrossNS  ,draw opacity=1 ] [dash pattern={on 4.5pt off 4.5pt}]  (344,150.08) -- (299.11,183.39) ;
  \draw [shift={(297.5,184.58)}, rotate = 323.43] [color=colorCrossNS  ,draw opacity=1 ][line width=0.75]    (10.93,-3.29) .. controls (6.95,-1.4) and (3.31,-0.3) .. (0,0) .. controls (3.31,0.3) and (6.95,1.4) .. (10.93,3.29)   ;
  \draw [color=colorCrossWE  ,draw opacity=1 ]   (304,214.58) -- (344.95,242.7) ;
  \draw [shift={(346.6,243.83)}, rotate = 214.47] [color=colorCrossWE  ,draw opacity=1 ][line width=0.75]    (10.93,-3.29) .. controls (6.95,-1.4) and (3.31,-0.3) .. (0,0) .. controls (3.31,0.3) and (6.95,1.4) .. (10.93,3.29)   ;
  \draw [color=colorCrossWE  ,draw opacity=1 ]   (346.6,243.83) -- (387.55,271.95) ;
  \draw [shift={(389.2,273.08)}, rotate = 214.47] [color=colorCrossWE  ,draw opacity=1 ][line width=0.75]    (10.93,-3.29) .. controls (6.95,-1.4) and (3.31,-0.3) .. (0,0) .. controls (3.31,0.3) and (6.95,1.4) .. (10.93,3.29)   ;
  \draw [color=colorCrossWE  ,draw opacity=1 ]   (218.8,156.08) -- (259.75,184.2) ;
  \draw [shift={(261.4,185.33)}, rotate = 214.47] [color=colorCrossWE  ,draw opacity=1 ][line width=0.75]    (10.93,-3.29) .. controls (6.95,-1.4) and (3.31,-0.3) .. (0,0) .. controls (3.31,0.3) and (6.95,1.4) .. (10.93,3.29)   ;

  \draw (95.17,199.17) node [anchor=north west][inner sep=0.75pt]   [align=left] {\textbf{{\LARGE W}}};
  \draw (503.5,275.83) node [anchor=north west][inner sep=0.75pt]   [align=left] {\textbf{{\LARGE E}}};
  \draw (254.67,310.17) node [anchor=north west][inner sep=0.75pt]   [align=left] {\textbf{{\LARGE S}}};
  \draw (320.5,84) node [anchor=north west][inner sep=0.75pt]   [align=left] {\textbf{{\LARGE N}}};
  \draw (147.5,155.83) node [anchor=north west][inner sep=0.75pt]   [align=left] {$\displaystyle a$};
  \draw (422.17,259.17) node [anchor=north west][inner sep=0.75pt]   [align=left] {$\displaystyle b$};
  \draw (417,110.67) node [anchor=north west][inner sep=0.75pt]   [align=left] {$\displaystyle c$};
  \draw (213.67,290.33) node [anchor=north west][inner sep=0.75pt]   [align=left] {$\displaystyle d$};
  \draw (452,178) node [anchor=north west][inner sep=0.75pt]  [color=colorCrossNS  ,opacity=1 ] [align=left] {$\displaystyle C_{1}$};
  \draw (140,235) node [anchor=north west][inner sep=0.75pt]  [color=colorCrossNS  ,opacity=1 ] [align=left] {$\displaystyle \textcolor[rgb]{0.82,0.01,0.11}{C}\textcolor[rgb]{0.82,0.01,0.11}{_{2}}$};
  \draw (342,253) node [anchor=north west][inner sep=0.75pt]  [color=colorCrossWE  ,opacity=1 ] [align=left] {$\displaystyle \vec{U}$};
  \draw (310,140) node [anchor=north west][inner sep=0.75pt]  [color=colorCrossNS  ,opacity=1 ] [align=left] {$\displaystyle \vec{V}$};

\end{tikzpicture}
  }
  \caption{A generic non-uniform crossing (left).
  The crossing after repeating the crossing directions (right).}
  \label{img:p2w-proof}
\end{figure}
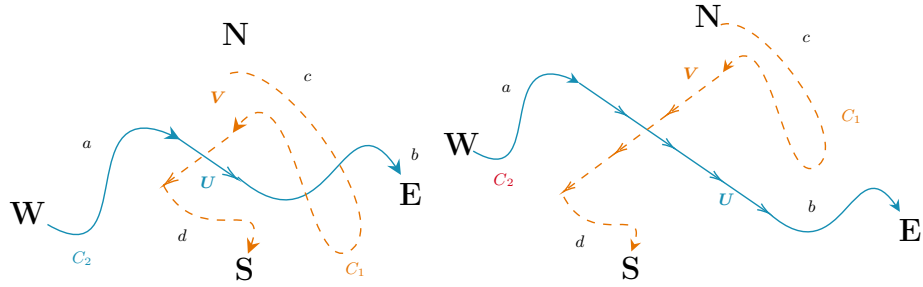

\begin{theorem}\label{thm:primal_iff_weighted}
  Let $\prim \subseteq \mathbb{Z}^2$ be a finite set. The following statements are equivalent:
  \begin{enumerate}
    \item A $\prim$-primal grid admits a crossing gate.
    \item There exists a weighted $\prim$-grid that admits a crossing gate.
  \end{enumerate}
\end{theorem}

\begin{proof}
  We first prove that \textit{(2)$\implies$(1)} and then that \textit{(1)$\implies$(2)}.

  \paragraph{(2)$\implies$(1)} We assume that there is a crossing gate $c$ on
  a weighted $\prim$-grid. Therefore, there exists
  a path in the North-toppling graph that transport grains from the
  North to the South, and analogously, there is a path from West to
  East in the West-graph. As uniform grids are Eulerian graphs, we can
  assume that these paths are disjoint \cite{nguyen_any_2018}.
  Therefore, we can construct a crossing gate for a
  $\prim$-primal grid by taking the directions given by the
  path, and taking an arbitrary out-neighbor from $\prim$ for all the other cells
  (they remain stable).

  \paragraph{(1)$\implies$(2)} We assume that there is a crossing
  gate $c$ on a $\prim$-primal grid. If we make a naive conversion
  from the $\prim$-primal grid to a $\prim$-grid, we may trigger 
  undesired topplings. However, we will show that by weighting
  well-chosen directions, we can construct a crossing gate.

  As $c$ is a crossing gate, we can consider a toppling path $C_1$
  from the North to the South and a path $C_2$ from the West to the
  East. We can assume them disjoints. We know that these two paths must
  cross at some point. We consider $U,V \in \prim$ the directions 
  that cross (Figure~\ref{img:p2w-proof} left). 

  We proceed as follows. First, we want to remove the possible interactions
  between $a, b$ and $c, d$. We can repeat $U$ and $V$ to move $a,b$ and $c,d$ away such that no 
  in/out-neighbor of both $a$ and $b$ is in $c$ or $d$
  (Figure~\ref{img:p2w-proof} right).
  If this created self-crossings of wires because the initial layout was spiraling,
  we trivially remove any extra loop.

  Then, we will use the fact that at the crossing, the paths are ``regular''. We
  remark that, edges from the repeated path $U$ to the repeated path $V$ (and conversly) cannot be
  done by the direction $U$ nor $B$. Therefore, by putting a big enough weight on
  directions $U$ and $V$, we can ensure that the toppling of the North path will not trigger the West
  path, and vice-versa. Only the directions $U$ and $V$ may interact and topple cells which
  guarantee the independence of the new crossing gate on the
  weighted $\prim$-grid.

  Finally, if one of the crossing direction is going in the ``wrong direction'', e.g.
  an upward edge for the North to South wire, then we need to change $c$ and $d$ (or $a,b$)
  to glue it back to the sides (because of our definition of crossing). Because the repetition
  may move away $c$ from the North side, we may need to change the direction of the crossing
  (e.g from South to North) to create the crossing gate.
  \qed
\end{proof}

\begin{example}
  As the von Neumann neighborhood is planar, none of its weightings
  admit a crossing gate. However, the Moore neighborhood is non-planar and
  some of its weightings do admit a crossing gate, as shown in
  Figure~\ref{fig:crossing-moore}.
\end{example}

\begin{figure}[t]
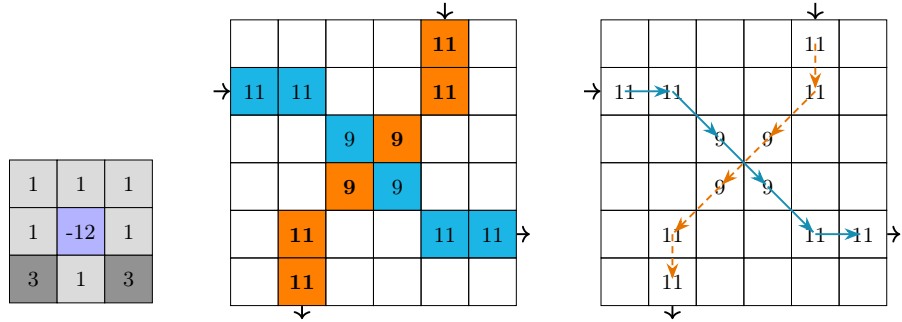

  \centering
  \scalebox{0.9}{
  \begin{tikzpicture}
    \begin{scope}[shift={(-4.3cm,-1.01cm)}]
      \input{tikzifieur/inter-moore-neig.tex}
    \end{scope}  
    \input{tikzifieur/inter-moore-crossing.tex}
  \end{tikzpicture}
  }
  \caption{A weighted Moore neighborhood (left) and a crossing gate (center)
  with its toppling graphs (right).}
  \label{fig:crossing-moore}
\end{figure}

We see that the proof only considers crossing gates with a unique edge-crossing.
This is because, as we can choose the weighting, we can increase the
weights of the two crossing directions in order to make them independent of the
other neighbors (as used in the proof). We will see in the next
section that gates with more complicated crossing exist.

\section{Intertwined crossing gates}\label{sec:imbricated}

Consider the neighborhood $\nei_i = \{(1,1), (-1,-1), (0,-1), (1,-1), (0,-2)\}$
(Figure~\ref{fig:inter-nei-named}). It admits a crossing
gate (Figure~\ref{fig:inter-nei-cross}). However, contrary to the
previous section, the two wires cross
with multiple steps. We argue that for every possible crossing gate
on a uniform grid with neighborhood $\nei_i$, there will always be a
cell in the middle of the crossing. Formally, we say that two edges
cross if, by drawing segments between the cell center, the segments (including the center endpoints) intersect.
We consider the end-extremity to be part of the seggment but not the start-extremity
(to avoid counting the vertical edge twice).

For the proof of the next result, we use the following notation.
Let $a,b \in V$ be two cells, $\nei$ a neighborhood,
$d \in \nei$ a \emph{direction}, and $c$ a crossing gate for the
$\nei$-grid with toppling graphs $G_n$ and $G_w$.
We write $a \neidir{d} b$ to indicate that $a + d = b$ and that
$(a,b)$ is an edge of $G_n$ or $G_w$, depending on the context.

\begin{lemma}
  For every crossing gate on a $\nei_i$-grid,
  there is a path in the toppling graph from the North to
  the South that crosses the
  toppling graph from West to East at least twice.
\end{lemma}

\begin{figure}[t]
  \centering
  \begin{tikzpicture}
    \input{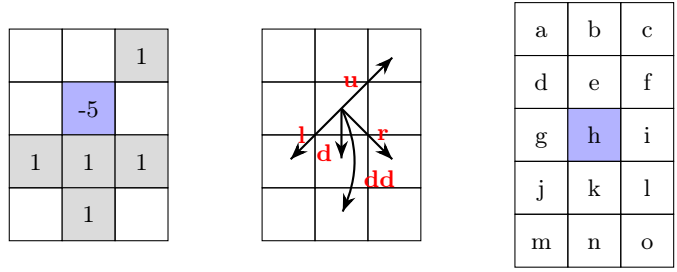}
    \matrix[table,right = of neig-inter-named-firing] {
      \node[] {a}; & \node[] {b}; & \node[] {c};\\
      \node[] {d}; & \node[] {e}; & \node[] {f};\\
      \node[] {g}; & \node[fill=blue!30] {h}; & \node[] {i};\\
      \node[] {j}; & \node[] {k}; & \node[] {l};\\
      \node[] {m}; & \node[] {n}; & \node[] {o};\\
    };
  \end{tikzpicture}
  \caption{Neighborhood $\nei_i$ with surrounding cells and direction names.}
  \label{fig:inter-nei-named}
\end{figure}
\vspace{-3mm}
\begin{figure}[t]
  \centering
  \scalebox{0.75}{
    \begin{tikzpicture}
      \input{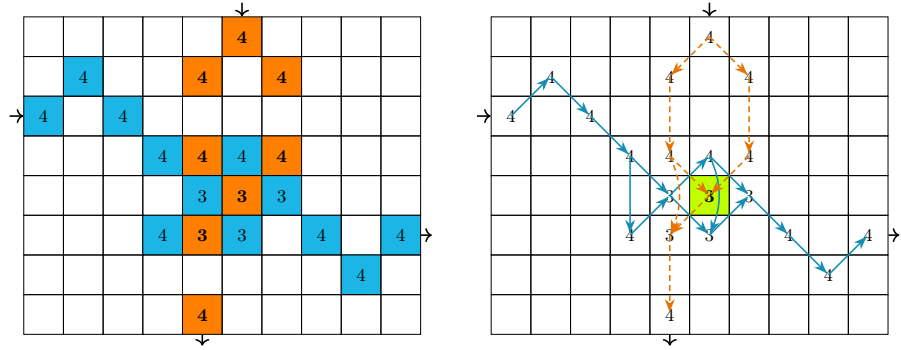}
    \end{tikzpicture}
  }
  \caption{Neighborhood $\nei_i$ with an intertwined crossing gate.}
  \label{fig:inter-nei-cross}
\end{figure}

\begin{proof}
  We name the directions and cells as in Figure~\ref{fig:inter-nei-named}.
  We proceed by contradiction. Let us assume that there exists a
  crossing gate with exactly one pair of crossing edges.
  We assume that all cells in the North-toppling graph $G_n$ contribute
  to the crossing (co-accessible from the South and accessible from the North).
  We also assume that $G_n$ and $G_w$ have disjoint vertex sets \cite{nguyen_any_2018}.
  Let us consider a minimal $h$ (for the distance from the North
  side) from the vertex set of $G_n$ which has at least one incoming edge crossing $G_w$,
  i.e., all the predecessors of $h$ have no incoming crossing edge.

  \smallskip
  \textbf{Claim.} \textit{$e$ belongs to the vertex set of $G_w$.}
  \smallskip

  As $h$ topples, we do a case analysis on the incoming arc that
  crosses $G_w$. 
  
  \begin{itemize}
    \item If $j \neidir{u} h$, the only crossing directions are $e \neidir{dd} k$,
          and $g \neidir{r} k$.
          Due the shape of $\nei_i$, cells will always topple to the South-East
          i.e., $q(v) := v_x + v_y$ is non-decreasing (if the origin $(0, 0)$ is the 
          North-West corner).

          Given that $g$ is accessible from the West,
          there must exists a path of cells $p_i$,
          from the West to $g$. We have $q(p_i) \leq q(g) < q(j)$.
          As $j$ is accessible from the North,
          there is a path in $G_n$ which must intersect $p_i$ at
          some point, which contradicts the assumption on $h$.

    \item If $e \neidir{d} h$, there is no valid edge that could cross it.
    
    \item If $b \neidir{dd} h$, all crossing directions pass through $e$.
    
    \item If $d \neidir{r} h$, the crossings directions are $e \neidir{dd} k$,
          $g \neidir{u} e$, and $e \neidir{l} g$, in all cases
          the claim on $e$ holds.

    \item If $f \neidir{l} h$, the crossings directions are $e \neidir{dd} k$ and
          $e \neidir{r} i$, in both cases the claim on $e$ holds.
  \end{itemize}

  We conclude that the claim always holds.
  Because the two toppling graphs are disjoint,
  $h$ must have two incoming arcs in $G_n$ because it
  can receive one grain from $e$. We now perform a case analysis
  on the incoming arcs of $h$.
  
  \begin{itemize}
    \item Directions $j \neidir{u} h$, $e \neidir{d} h$ are still not possible
          for the same reasons (minimality of $h$ and no valid crossing).

    \item If $b \neidir{dd} h$, then $e$ must have two incoming edges because $b \neidir{d} e$.
          Therefore, there are two incoming edges that intersect $b \neidir{dd} h$.

    \item Otherwise, $d \neidir{r} h$ and $f \neidir{l} h$, therefore, by arguments
          similar to the $j \neidir{u} h$ case, and because $e$ is accessible and
          contributes to the crossing, it must cross $G_n$ twice: $d \neidir{r} h$, then 
          $f \neidir{l} h$ (by hypothesis, it cannot cross predecessors on $h$).

  \end{itemize}\qed
\end{proof}

\section{Counter-examples to general decomposition}\label{sec:general-decomposition}

As shown in Section~\ref{sec:primal_iff_weighted}, non-uniformity and weighting
can be traded. However, as demonstrated in Section~\ref{sec:imbricated},
this equivalence only captures simple crossings.
We attempted to generalize the notion of primal-grid into a broader decomposition.
This exploration led to the identification of counter-examples demonstrating that
crossing gates are not easily composable.
In this section we exhibit those counter-examples.

To generalize the previous definition of primal grids, we consider decompositions
of arbitrary shape rather than restricting to primal directions.
Let $S$ be a set of (possibly weighted) neighborhoods, we are interested in the relation between non-uniform grid composed of neighborhoods of $S$,
and uniform grids whose neighborhoods are multisets on $S$.
Our results indicate that the existence of crossing gates
in one setting does not imply their existence in the other,
revealing fundamental differences between these grid constructions

Our techniques are applicable only to grids satisfying the following property:
adding one on one edge, each cell topples at most once during stabilization.
This property does not hold in general for non-uniform grids.
To circumvent this limitation, in this section we consider a variant of the
sandpile dynamics in which each cell is allowed to topple at most once.

\begin{figure}[t]
  \centering
  \scalebox{0.9}{
    \begin{tikzpicture}
      \input{tikzifieur/neig-nonuni-wei.tex}
    \end{tikzpicture}
  }
  \caption{Examples of neighborhoods admitting a non-uniform crossing gate but 
           whose weighted grids does not}
  \label{fig:nonuni-to-wei}
\end{figure}

\begin{example}[Figure~\ref{fig:nonuni-to-wei}]
  \label{ex:non-uni-a}
  Consider $S = \{\nei_{a1}, \nei_{a2}\}$. There is a crossing gate on a
  non-uniform grid composed of $\nei_{a1}$ and $\nei_{a2}$: one can use $\nei_{a1}$
  to offset the y-coordinate parity, and then use $\nei_{a2}$ to cross without
  interference.

  However, a $\{\nei_{a1}, \nei_{a2}\}$ uniform grid (a grid
  whoose neighborhood is a linear combination of $\nei_{a1}$ and $\nei_{a2}$)
  does not admit a crossing gate \cite{concha2025sandpiles}.
  A uniform grid with weighted neighborhood $\{\nei_{a1}, \nei_{a2}, \nei_{a2}\}$
  does admit a crossing gate, but it is not the case in general.
\end{example}

\begin{example}[Figure~\ref{fig:nonuni-to-wei}]
  Consider $S = \{\nei_{b1}, \nei_{b2}\}$, a non-uniform $S$ grid admits a crossing,
  by a path of $\nei_{b1}$ crossing a path of $\nei_{b2}$.
  However, any uniform grid with neighborhood $\nei_{b1}^p \nei_{b2}^q$ 
  ($p$ copies of $\nei_{b1}$ and $q$ copies of $\nei_{b2}$)
  does not admit a crossing gate.

  Arguing by contradiction, suppose that such a crossing gate exists.
  As observed previously, the corresponding toppling graphs must cross
  at some point, and only the diagonal directions can cross.
  If $p \neq q$, then independence of the toppling graphs is not verified
  because the toppling process on the heavier side necessarily triggers
  the toppling of the lighter one.
  On the other hand, Concha-Vega et al.~\cite{concha2025sandpiles}
  proved that the case $p = q$ does not admit a crossing gate.
\end{example}

\begin{figure}[t]
  \centering
  \scalebox{0.9}{
    \begin{tikzpicture}
      \input{tikzifieur/neig-wei-nonuni.tex}
    \end{tikzpicture}
  }
  \caption{Examples of neighborhoods admitting a weighted crossing gate but 
           whose non-uniform grids does not}
  \label{fig:wei-to-nonui}
\end{figure}

\begin{example}[Figure~\ref{fig:wei-to-nonui}]
  The uniform grid with neighborhood $\{\nei_{c1}, \nei_{c2}\}$
  admits a crossing gate (for example a diagonal cross).
  On the other hand, any $\{\nei_{c1}, \nei_{c2}\}$ non-uniform grid
  does not admit a crossing gate.

  Arguing by contradiction, suppose that such a crossing gate exists.
  We assume again that the two topplings are disjoint.
  The only possible crossings are between the diagonal direction $\nei_{c2}$
  and the 2-weighted direction in $\nei_{c1}$.
  Let us call $j$ the out-neighbor according to diagonal direction $\nei_{c2}$
  in this crossing. Then the direction $\nei_{c2}$ sends
  one grain to cell $j$: this is an arc in one of the two topping graphs,
  and $j$ has no other in-neighbor.
  However, cell $j$ also receives one grain from the cell above,
  and this cell is in the other toppling graph, therefore contradicting
  disjointness because it topples $j$.
\end{example}

\section{Conclusion and perspectives}\label{sec:conclusions}

To conclude, in section~\ref{sec:primal_iff_weighted},
we established a characterization of weighted uniform crossing gates
in terms of a simpler class of structures, which we call primal crossings.
However, this characterization does not capture more complex crossings, like the
intertwined crossings of section~\ref{sec:imbricated}.
Finally, in Section~\ref{sec:general-decomposition}, we show that the above
equivalence does not extend to more general settings,
where the decomposition into simpler but not primal structures fails.
We emphasize that Theorem~\ref{thm:primal_iff_weighted} facilitates
$\P$-completeness proofs, as the construction of the remaining components
for reductions from \textbf{CVP} is typically straightforward.

While crossing gates are a central ingredient in proving \P-completeness
for cellular automata, their absence does not rule out \P-completeness,
since there may exist other ways of encoding signals than just creating
chains of topplings or reductions may exist beyond chain-like toppling structures.

Moreover, even when the equivalence preserves the existence of crossing gate,
it does not preserves the computational complexity:
on primal grids, grain movement is highly restricted
allowing efficient \NC algorithms for the prediction problem.
In contrast, weighted grids appear to support significantly richer dynamics
and are natural candidates for \P-completeness.

We leave open what we consider our Holy Grail: the complexity of prediction
in 2-dimensional uniform sandpiles with von Neumann or Moore neighborhoods.
More broadly, this work raises structural questions about crossing gates:
in particular, it is unclear whether there exist families of grids requiring
arbitrarily large numbers of crossing edges, or whether deciding the
existence of a crossing gate is even decidable in general.

\begin{credits}
  \subsubsection{\ackname}
  The authors thank projects
  ANR-24-CE48-7504 ALARICE and
  HORIZON-MSCA-2022-SE-01 101131549 ACANCOS.
  \subsubsection{\discintname}
  The authors have no competing interests to declare.
\end{credits}
%
%
%
\bibliographystyle{splncs04}
\bibliography{biblio.bib}

\end{document}